\documentclass[english, letterpaper, 11pt]{article}

\usepackage{setspace}
\usepackage[T1]{fontenc}
\usepackage[utf8]{inputenc}
\usepackage{babel}
\usepackage{blindtext}
\usepackage{cite}
\usepackage[margin=0.5in]{geometry}
\usepackage[margin=0.5in]{caption}
\usepackage[pdftex]{graphicx}
\graphicspath{{./fig/}}
\usepackage[cmex10]{amsmath}
\usepackage{amssymb}
\usepackage{array}
\usepackage[tight,footnotesize]{subfigure}
\usepackage{stfloats}
\usepackage{algorithm,algorithmic}
\usepackage{enumerate}
\usepackage{amsthm}
\usepackage{color}
\usepackage{mathrsfs}
\usepackage{multicol}
\usepackage{mathtools}
\usepackage{tikz}
\usepackage{pgfplots}
\usepackage{umoline}
\usepackage{arydshln}

\providecommand{\keywords}[1]{\bigskip\textbf{\textit{Index terms---}} #1}

\let\emptyset\varnothing
\newcommand{\eqdef}{\vcentcolon=}

\newcommand{\diag}{\operatorname{diag}}

\newcommand{\spanof}{\operatorname{span}}

\newcommand{\bbC}{\mathbb{C}}

\newcommand{\calX}{\mathcal{X}}
\newcommand{\calY}{\mathcal{Y}}

\newcommand{\calM}{\mathcal{M}}
\newcommand{\calN}{\mathcal{N}}
\newcommand{\calR}{\mathcal{R}}
\newcommand{\calU}{\mathcal{U}}
\newcommand{\calV}{\mathcal{V}}
\newcommand{\calG}{\mathcal{G}}

\newtheorem{theorem}{Theorem}[section]
\newtheorem{lemma}[theorem]{Lemma}
\newtheorem{corollary}[theorem]{Corollary}
\newtheorem{proposition}[theorem]{Proposition}
\newtheorem{definition}[theorem]{Definition}

\begin{document}
\title{Identifiability in Blind Deconvolution with Subspace or Sparsity Constraints\thanks{This work was supported in part by the National Science Foundation (NSF) under Grants CCF 10-18789 and IIS 14-47879. Some of these results will be presented at SPARS 2015 \cite{Li2015c}}}
\author{Yanjun Li\thanks{Department of Electrical and Computer Engineering and Coordinated Science Laboratory, University of Illinois, Urbana-Champaign, IL 61801, USA.}
\and
Kiryung Lee\thanks{Department of Statistics and Coordinated Science Laboratory, University of Illinois, Urbana-Champaign, IL 61801, USA.}
\and
Yoram Bresler\footnotemark[2]}
\maketitle

\doublespacing

\abstract
Blind deconvolution (BD), the resolution of a signal and a filter given their convolution, arises in many applications. Without further constraints, BD is ill-posed. In practice, subspace or sparsity constraints have been imposed to reduce the search space, and have shown some empirical success. However, existing theoretical analysis on uniqueness in BD is rather limited. As an effort to address the still mysterious question, we derive sufficient conditions under which two vectors can be uniquely identified from their circular convolution, subject to subspace or sparsity constraints. These sufficient conditions provide the first algebraic sample complexities for BD. We first derive a sufficient condition that applies to almost all bases or frames. For blind deconvolution of vectors in $\bbC^n$, with two subspace constraints of dimensions $m_1$ and $m_2$, the required sample complexity is $n\geq m_1m_2$. Then we impose a sub-band structure on one basis, and derive a sufficient condition that involves a relaxed sample complexity $n\geq m_1+m_2-1$, which we show to be optimal. We present the extensions of these results to BD with sparsity constraints or mixed constraints, with the sparsity level replacing the subspace dimension. The cost for the unknown support in this case is an extra factor of $2$ in the sample complexity.


\section{Introduction}
Blind deconvolution (BD) is the bilinear inverse problem of recovering the signal and the filter simultaneously given the their convolutioin or circular convolution. It arises in many applications, including blind image deblurring \cite{Kundur1996}, blind channel equalization \cite{LitwinL.R.1999}, speech dereverberation \cite{Naylor2010}, and seismic data analysis \cite{Yilmaz2001}. Without further constraints, BD is an ill-posed problem, and does not yield a unique solution. A variety of constraints have been introduced to exploit the properties of natural signals and reduce the search space. Examples of such constraints include positivity (the signals are non-negative), subspace constraint (the signals reside in a lower-dimensional subspace) and sparsity (the signals are sparse over some dictionary). In this paper, we focus on subspace or sparsity constraints, which can be imposed on both the signal and the filter. Consider the example of blind image deblurring: a natural image can be considered sparse over the a wavelet dictionary or the discrete cosine transform (DCT) dictionary. The support of the point spread function (PSF) is usually significantly smaller than the image itself. Therefore the filter resides in a lower-dimensional subspace. These priors serve as constraints or regularizers \cite{Chan1998,Herrity2008b,Krishnan2011,Asif2009,Ahmed2014}. With a reduced search space, BD can be better-posed. However, despite the success in practice, the theoretical results on the uniqueness in BD with a subspace or sparsity constraint are limited.

Early works on the identifiability in blind deconvolution studied multichannel blind deconvolution with finite impulse response (FIR) models \cite{Moulines1995,Abed-meraim1997}, in which sparsity was not considered. For single channel blind deconvolution, sparsity was imposed as a prior without theoretical justification \cite{Chan1998,Herrity2008b,Asif2009,Krishnan2011,Repetti2015}.

Recently, recasting bilinear or quadratic inverse problems, such as blind deconvolution \cite{Ahmed2014} and phase retrieval \cite{Candes2013a}, as rank-$1$ matrix recovery problems by ``lifting'' has attracted a lot of attention. 
Choudhary and Mitra \cite{Choudhary2013a} adopted the lifting framework and showed that the identifiability in BD (or any bilinear inverse problem) hinges on the set of rank-$2$ matrices in a certain nullspace. In particular, they showed a negative result that the solution to blind deconvolution with a canonical sparsity prior, that is, sparsity over the natural basis, is \emph{not} identifiable \cite{Choudhary2014a}. However, the identifiability of signals that are sparse over other dictionaries has not been analyzed.

Using the lifting framework, Ahmed et al. \cite{Ahmed2014} showed that BD with subspace constraints is identifiable up to scaling. More specifically, if the signal subspace follows a random Gaussian model, and the filter subspace satisfies some coherence conditions, convex programming was shown to recover the signal and the filter up to scaling with high probability, when the dimensions of the subspaces $m_1$ and $m_2$ are in a near optimal regime $m_1+m_2 = O(n)$, where $n$ denotes the length of the signal. Ling and Strohmer \cite{Ling2015} extended the model in \cite{Ahmed2014} to blind deconvolution with mixed constraints: the signal is sparse over a random Gaussian dictionary or a randomly subsampled partial Fourier matrix, and the filter resides in a subspace that satisfies some coherence condition. They showed that the signal and the filter can be simultaneously identified with high probability using $\ell_1$ norm minimization (instead of nuclear norm minimization as in \cite{Ahmed2014}) when the sparsity level $s_1$ and the subspace dimension $m_2$ satisfy $s_1m_2=O(n)$. Lee et al. \cite{Lee2015} further extended the model to blind deconvolution with sparsity constraints on both the signal and the filter, and showed successful recovery with high probability using alternating minimization when the sparsity levels $s_1$ and $s_2$ satisfy $s_1+s_2=O(n)$. A common drawback of these works is that the probabilistic assumptions on the bases or frames are very limiting in practice. On the positive side, these identifiability results are constructive, being demonstrated by establishing performance guarantees of algorithms. However, these guarantees too are shown only in some probabilistic sense.

To overcome the limitations of the lifting framework, the authors of this paper introduced a more general framework for the identifiability in bilinear inverse problems \cite{Li2015}, namely, identifiability up to transformation groups. We showed that two vectors $x,y$ are identifiable up to a transformation group given their image under a bilinear mapping if:
\begin{enumerate}
	\item $x$ is identifiable up to the transformation group;
	\item once $x$ is identified, the recovery of $y$ is unique.
\end{enumerate}
For multichannel blind deconvolution, we were able to derive identifiability results \cite{Li2015} under subspace, joint sparsity or sparsity constraints within our framework.

In this paper, we address the identifiability in single channel blind deconvolution up to scaling under subspace or sparsity constraints. We present the first algebraic sample complexities for BD with fully deterministic signal models. The rest of the paper is organized as follows. We formally state the problem setup in Section \ref{sec:probstat}. In Section \ref{sec:gebf}, we derive sufficient conditions for BD with generic bases or frames, using the lifting framework. In Section \ref{sec:sbsb}, we derive much less demanding sufficient conditions for BD with a sub-band structured basis, using the framework in \cite{Li2015}. Notably, the sample complexities of the sufficient conditions in this case match those of corresponding necessary conditions, and hence are optimal. We conclude the paper in Section \ref{sec:conclusions} with some remarks and open questions.

\section{Problem Statement}\label{sec:probstat}
\subsection{Notations}
We state the notations that will be used throughout the paper. We use lower-case letters $x$, $y$, $z$ to denote vectors, and upper-case letters $D$ and $E$ to denote matrices. We use $I_n$ to denote the identity matrix and $F_n$ to denote the normalized discrete Fourier transform (DFT) matrix. The DFT of the vector $x\in\bbC^n$ is denoted by $\widetilde{x}=F_nx$. We use $\mathbf{1}_{m,n}$ to denote a matrix whose entries are all ones and $\mathbf{0}_{m,n}$ to denote a matrix whose entries are all zeros. The subscripts stand for the dimensions of these matrices. We say that a vector is \emph{non-vanishing} if all its entries are nonzero. Unless otherwise stated, all vectors are column vectors. The dimensions of all vectors and matrices are made clear in the context. 

The projection operator onto a subspace $\calV$ is denoted by $P_\calV$.
The nullspace and the range space of a linear operator are denoted by $\calN(\cdot)$ and $\calR(\cdot)$, respectively. 
We use $\Omega_\calX,\Omega_\calY$ to denote constraint sets. The Cartesian product of two sets are denoted by $\Omega_\calX\times \Omega_\calY$. The pair $(x,y)\in \Omega_\calX\times \Omega_\calY$ represents an element of the Cartesian product. 
We use $./$ and $\odot$ to denote entrywise division and entrywise product, respectively. Circular convolution is denoted by $\circledast$. Kronecker product is denoted by $\otimes$. The direct sum of two subspaces is denoted by $\oplus$.

We use $j,k$ to denote indices, and $J,K$ to denote index sets. If the universal index set is $\{1,2,\cdots,n\}$, then $J,K$ are subsets. We use $|J|$ to denote the cardinality of $J$. We use $J^c$ to denote the complement of $J$. We use superscript letters to denote subvectors or submatrices. For example, $x^{(J)}$ represents the subvector of $x$ consisting of the entries indexed by $J$. The scalar $x^{(j)}$ represents the $j$th entry of $x$. The submatrix $D^{(J,K)}$ has size $|J|\times |K|$ and consists of the entries indexed by $J\times K$. The vector $D^{(:,k)}$ represents the $k$th column of the matrix $D$. The colon notation is borrowed from MATLAB.

We say a property holds for almost all signals (generic signals) if the property holds for all signals but a set of measure zero.

\subsection{Blind Deconvolution}
In this paper, we study the blind deconvolution problem with the circular convolution model. It is the joint recovery of two vectors $u_0\in\bbC^n$ and $v_0\in\bbC^n$, namely the signal and the filter\footnote{Due to symmetry, the name ``signal'' and ``filter'' can be used interchangeably.}, given their circular convolution $z=u_0\circledast v_0$, subject to subspace or sparsity constraints. The constraint sets $\Omega_\calU$ and $\Omega_\calV$ are subsets of $\bbC^n$. 
\begin{align*}
\text{find}~~&(u,v),\\
\text{s.t.}~~&u\circledast v = z,\\
& u \in \Omega_\calU,~ v \in \Omega_\calV.
\end{align*}
We consider the following scenarios:
\begin{enumerate}
	\item \emph{(Subspace Constraints)} The signal $u$ and the filter $v$ reside in lower-dimensional subspaces spanned by the columns of $D\in \bbC^{n\times m_1}$ and $E\in \bbC^{n\times m_2}$, respectively. The matrices $D$ and $E$ have full column ranks. Therefore,
	\begin{align*}
	\Omega_\calU &= \left\{u\in\bbC^n: u=Dx \text{ for some } x\in\bbC^{m_1}\right\}, \\
	\Omega_\calV &= \left\{v\in\bbC^n: v=Ey \text{ for some } y\in\bbC^{m_2}\right\}.
	\end{align*}
	\item \emph{(Sparsity Constraints)} The signal $u$ and the filter $v$ are sparse over given dictionaries formed by the columns of $D\in \bbC^{n\times m_1}$ and $E\in \bbC^{n\times m_2}$, with sparsity level $s_1$ and $s_2$, respectively. The matrices $D$ and $E$ are bases or frames that satisfy the spark condition \cite{Donoho2003}: the spark, namely the smallest number of columns that are linearly dependent, of $D$ (resp. $E$) is greater than $2s_1$ (resp. $2s_2$).
	Therefore,
	\begin{align*}
	\Omega_\calU &= \left\{u\in\bbC^n: u=Dx \text{ for some } x\in\bbC^{m_1}\text{ s.t. }\|x\|_0\leq s_1 \right\}, \\
	\Omega_\calV &= \left\{v\in\bbC^n: v=Ey \text{ for some } y\in\bbC^{m_2}\text{ s.t. }\|y\|_0\leq s_2 \right\}.
	\end{align*}
	\item \emph{(Mixed Constraints)} The signal $u$ is sparse over a given dictionary $D\in \bbC^{n\times m_1}$, and the filter $v$ resides in a lower-dimensional subspace spanned by the columns of $E\in \bbC^{n\times m_2}$. The matrix $D$ satisfies the spark condition, and $E$ has full column rank. Therefore,
	\begin{align*}
	\Omega_\calU &= \left\{u\in\bbC^n: u=Dx \text{ for some } x\in\bbC^{m_1}\text{ s.t. }\|x\|_0\leq s_1 \right\}, \\
	\Omega_\calV &= \left\{v\in\bbC^n: v=Ey \text{ for some } y\in\bbC^{m_2}\right\}.
	\end{align*}
\end{enumerate}

In all three scenarios, the vectors $x$, $y$, and $z$ reside in Euclidean spaces $\bbC^{m_1}$, $\bbC^{m_2}$ and $\bbC^{n}$. With the representations $u=Dx$ and $v=Ey$, it is easy to verify that $z=u\circledast v=(Dx)\circledast(Ey)$ is a bilinear function of $x$ and $y$. 
Given the measurement $z=(Dx_0)\circledast(Ey_0)$, the blind deconvolution problem can be rewritten in the following form:
\begin{align*}
\text{(BD)}\qquad\text{find}~~&(x,y)\\
\text{s.t.}~~&(Dx)\circledast (Ey) = z\\
& x \in \Omega_\calX,~ y \in \Omega_\calY
\end{align*}
If $D$ and $E$ satisfy the full column rank condition or spark condition, then 
the uniqueness of $(u,v)$ is equivalent to the uniqueness of $(x,y)$.
For simplicity, we will discuss problem (BD) from now on. The constraint sets $\Omega_\calX$ and $\Omega_\calY$ depend on the constraints on the signal and the filter. For subspace constraints, $\Omega_\calX=\bbC^{m_1}$, $\Omega_\calY = \bbC^{m_2}$. For sparsity constraints, $\Omega_\calX =\{x\in \bbC^{m_1}: \|x\|_0\leq s_1\}$, $\Omega_\calY = \{y\in \bbC^{m_2}: \|y\|_0\leq s_2\}$.

\subsection{Identifiability up to Scaling}
An important question concerning the blind deconvolution problem is to determine when it admits a unique solution. The BD problem suffers from scaling ambiguity.
For any nonzero scalar $\sigma \in \bbC$ such that $\sigma x_0 \in \Omega_\calX$ and $\frac{1}{\sigma} y_0 \in \Omega_\calY$, $(D(\sigma x_0))\circledast (E(\frac{1}{\sigma} y_0))=(Dx_0)\circledast (Ey_0)=z$. Therefore, BD does not yield a unique solution if $\Omega_\calX,\Omega_\calY$ contain such scaled versions of $x_0,y_0$. 
Any valid definition of unique recovery in BD must address this issue. If every solution $(x,y)$ is a scaled version of $(x_0,y_0)$, then we must say $(x_0,y_0)$ can be uniquely identified up to scaling. We define identifiability as follows.

\begin{definition}\label{def:ibip}
For the constrained BD problem, the solution $(x_0,y_0)$, in which $x_0\neq 0$ and $y_0\neq 0$, is said to be identifiable up to scaling, if every solution $(x,y)\in \Omega_\calX\times \Omega_\calY$ satisfies $x=\sigma x_0$ and $y=\frac{1}{\sigma} y_0$. 
\end{definition}

For blind deconvolution, there exists a linear operator $\mathcal{G}_{DE}:\bbC^{m_1\times m_2}\rightarrow \bbC^n$ such that $\mathcal{G}_{DE}(xy^T)=(Dx)\circledast(Ey)$. Given the measurement $z=\mathcal{G}_{DE}(x_0y_0^T)=(Dx_0)\circledast(Ey_0)$, one can recast the BD problem as the recovery of the rank-$1$ matrix $M_0=x_0y_0^T \in \Omega_\calM=\{xy^T:x\in \Omega_\calX,y\in \Omega_\calY\}$. The uniqueness of $M_0$ is equivalent to the identifiability of $(x_0,y_0)$ up to scaling. This procedure is called ``lifting''.
\begin{align*}
\text{(Lifted BD)}\qquad\text{find}~~& M,\\
\text{s.t.}~~&\mathcal{G}_{DE}(M) = z,\\
& M\in \Omega_\calM.
\end{align*}
It was shown \cite{Choudhary2013a} that the lifted BD has a unique solution for every $M_0\in \Omega_\calM$ if the nullspace of $\mathcal{G}_{DE}$ does not contain the difference of two different matrices in $\Omega_\calM$. 
\begin{proposition}\label{pro:ilift}
The pair $(x_0,y_0)\in\Omega_\calX\times\Omega_\calY$ ($x_0\neq 0$, $y_0\neq 0$) is identifiable up to scaling in (BD), or equivalently, the solution $M_0=x_0y_0^T\in \Omega_\calM$ is unique in (Lifted BD), if and only if
\[
\calN(\mathcal{G}_{DE})\bigcap \{M_0-M:M\in \Omega_\calM\} = \{0\}.
\]
\end{proposition}

Proposition \ref{pro:ilift} is difficult to apply because it is not clear how to find the nullspace of the structured linear operator $\mathcal{G}_{DE}$. To overcome this limitation, in our previous work (see \cite{Li2015} Theorem 2.8), we derived a necessary and sufficient condition for the identifiability in a bilinear inverse problem up to a transformation group. As a special case, we have the following necessary and sufficient condition for the identifiability in BD up to scaling, which holds for any $\Omega_\calX$ and $\Omega_\calY$.
\begin{proposition}\label{pro:ibd}
The pair $(x_0,y_0)\in\Omega_\calX\times\Omega_\calY$ ($x_0\neq 0$, $y_0\neq 0$) is identifiable up to scaling in (BD) if and only if the following two conditions are met:
\begin{enumerate}
	\item If there exists $(x,y)\in\Omega_\calX\times \Omega_\calY$ such that $(Dx)\circledast (Ey)=(Dx_0)\circledast (Ey_0)$, then $x = \sigma x_0$ for some nonzero $\sigma\in\bbC$.
	\item If there exists $y\in\Omega_\calY$ such that $(Dx_0)\circledast (Ey)=(Dx_0)\circledast (Ey_0)$, then $y=y_0$.
\end{enumerate}
\end{proposition}

Propositions \ref{pro:ilift} and \ref{pro:ibd} are two equivalent conditions for the identifiability in blind deconvolution. Proposition \ref{pro:ilift} shows how the identifiability of $(x,y)$ is connected to that of the lifted variable $xy^T$. Proposition \ref{pro:ibd} shows how the identifiability of $(x,y)$ can be divided into the identifiability of $x$ and $y$ individually.
In this paper, we derive more readily interpretable conditions for the uniqueness of solution to BD with subspace or sparsity constraints. We first derive a sufficient condition for the case where the bases or frames are generic, using the lifting framework. We also apply \ref{pro:ibd} and derive a sufficient condition for the case where one of the bases has a sub-band structure.

\section{Blind Deconvolution with Generic Bases or Frames} \label{sec:gebf}
Subspace membership and sparsity have been used as priors in blind deconvolution for a long time. 
Previous works either use these priors without theoretical justification \cite{Chan1998,Herrity2008b,Asif2009,Krishnan2011,Repetti2015}, or impose probablistic models and show successful recovery with high probability \cite{Ahmed2014,Ling2015,Lee2015}. In this section, we derive sufficient conditions for the identifiability of blind deconvolution under subspace or sparsity constraints. These conditions are fully deterministic and provide uniform upper bounds for the sample complexities of blind deconvolution with almost all bases or frames. 

The identifiability of $(x_0,y_0)$ up to scaling in (BD) is equivalent to the uniqueness of $M_0=x_0y_0^T$ in (Lifted BD). The linear operator $\mathcal{G}_{DE}$ can also be represented by a matrix $G_{DE}\in\bbC^{n\times m_1m_2}$ such that $\mathcal{G}_{DE}(M_0)=G_{DE}\operatorname{vec}(M_0)$, where $\operatorname{vec}(M_0)$ stacks the columns of $M_0\in\bbC^{m_1\times m_2}$ on top of one another and forms a vector in $\bbC^{m_1m_2}$. The columns of $G_{DE}$ have the form $D^{(:,j)}\circledast E^{(:,k)}=\sqrt{n}F_n^*(\widetilde{D}^{(:,j)}\odot \widetilde{E}^{(:,k)})$, where $\widetilde{D}=F_nD$ and $\widetilde{E}=F_nE$. Clearly, the matrix $G_{DE}$ is a function of the matrices $D$ and $E$. It has the following properties (See Appendix \ref{app:fcr} for the proofs).

\begin{lemma}\label{lem:fcr1}
If $n\geq m_1m_2$, then for almost all $D\in\bbC^{n\times m_1}$ and $E\in\bbC^{n\times m_2}$, the matrix $G_{DE}$ has full column rank.
\end{lemma}

\begin{lemma}\label{lem:fcr2}
If $n\geq 2s_1m_2$, then for any $0\leq t_1\leq s_1$, for almost all $D_0\in\bbC^{n\times t_1}$, $D_1\in\bbC^{n\times (s_1-t_1)}$, $D_2\in\bbC^{n\times (s_1-t_1)}$, and $E\in\bbC^{n\times m_2}$, the columns of $G_{D_0E}$, $G_{D_1E}$, $G_{D_2E}$ together form a linearly independent set.
\end{lemma}

\begin{lemma}\label{lem:fcr3}
If $n\geq 2s_1s_2$, then for any $0\leq t_1\leq s_1$, $0\leq t_2\leq s_2$, for almost all $D_0\in\bbC^{n\times t_1}$, $D_1\in\bbC^{n\times (s_1-t_1)}$, $D_2\in\bbC^{n\times (s_1-t_1)}$, $E_0\in\bbC^{n\times t_2}$, $E_1\in\bbC^{n\times (s_2-t_2)}$, and $E_2\in\bbC^{n\times (s_2-t_2)}$, the columns of $G_{D_0E_0}$, $G_{D_1E_0}$, $G_{D_2E_0}$, $G_{D_0E_1}$, $G_{D_1E_1}$, $G_{D_0E_2}$, $G_{D_2E_2}$ together form a linearly independent set.
\end{lemma}

Next, we state and prove sufficient conditions for identifiability of blind deconvolution with generic bases or frames.

\begin{theorem}[Subspace Constraints]\label{thm:gebf1}
In (BD) with subspace constraints, $(x_0,y_0)\in\bbC^{m_1}\times \bbC^{m_2}$ ($x_0\neq 0$, $y_0\neq 0$) is identifiable up to scaling, for almost all $D\in\bbC^{n\times m_1}$ and $E\in\bbC^{n\times m_2}$, if $n\geq m_1m_2$.
\end{theorem}
\begin{proof}
By Lemma \ref{lem:fcr1}, if $n\geq m_1m_2$, for almost all $D\in\bbC^{n\times m_1}$ and $E\in\bbC^{n\times m_2}$, the matrix $G_{DE}$ has full column rank. Therefore, $\calN(\mathcal{G}_{DE})=\{0\}$, and the lifted problem has a unique solution. It follows that every pair $(x_0,y_0)$ is identifiable up to scaling.
\end{proof}

\begin{theorem}[Mixed Constraints]\label{thm:gebf2}
In (BD) with mixed constraints, $(x_0,y_0)\in\bbC^{m_1}\times \bbC^{m_2}$ ($\|x_0\|_0\leq s_1$, $x_0\neq 0$, $y_0\neq 0$) is identifiable up to scaling, for almost all $D\in\bbC^{n\times m_1}$ and $E\in\bbC^{n\times m_2}$, if $n\geq 2s_1m_2$.
\end{theorem}
\begin{proof}
Fix index sets $J_0,J\subset\{1,2,\cdots,m_1\}$, for which $|J_0|=|J|=s_1$ and $|J_0\bigcap J|=t_1$. Let 
\[
D_0 = D^{(:,J_0\bigcap J)}\in\bbC^{n\times t_1}, \quad D_1 = D^{(:,J_0\setminus J)}\in\bbC^{n\times (s_1-t_1)}, \quad D_2 = D^{(:,J\setminus J_0)}\in\bbC^{n\times (s_1-t_1)}.
\]
By Lemma \ref{lem:fcr2}, if $n\geq 2s_1m_2$, then for almost all $D$ and $E$, the columns of $G_{D_0E}$, $G_{D_1E}$, $G_{D_2E}$ together form a linearly independent set. For every $(x_0,y_0)$ and $(x,y)$ such that the $s_1$-sparse $x_0$ and $x$ are supported on $J_0$ and $J$ respectively, if $\mathcal{G}_{DE}(x_0y_0^T)=\mathcal{G}_{DE}(xy^T)$, then
\[
G_{D_0E}v_0+G_{D_1E}v_1+G_{D_2E}v_2=G_{DE}v=\mathcal{G}_{DE}(x_0y_0^T)-\mathcal{G}_{DE}(xy^T)=0,
\]
where $v=\operatorname{vec}(x_0y_0^T-xy^T)$, $v_0 = \operatorname{vec}(x_0^{(J_0\bigcap J)}y_0^T-x^{(J_0\bigcap J)}y^T)$, $v_1 = \operatorname{vec}(x_0^{(J_0\setminus J)}y_0^T)$ and $v_2 = \operatorname{vec}(-x^{(J\setminus J_0)}y^T)$. By linear independence, the vectors $v_0,v_1,v_2$ are all zero vectors, and so is $v$. Hence for almost all $D$ and $E$, and all pairs $(x_0,y_0)$ and $(x,y)$ such that $x_0$ and $x$ are supported on $J_0$ and $J$ respectively, if $\mathcal{G}_{DE}(x_0y_0^T)=\mathcal{G}_{DE}(xy^T)$, then $x_0y_0^T=xy^T$. Note that $J_0$ and $J$ are arbitrary, and there is only a finite number (${m_1 \choose s_1}^2$) of combinations of $J_0,J$. Therefore, for almost all $D$ and $E$, every pair $(x_0,y_0)$ ($\|x_0\|_0\leq s_1$, $x_0\neq 0$, $y_0\neq 0$) is identifiable up to scaling.
\end{proof}

\begin{theorem}[Sparsity Constraints]\label{thm:gebf3}
In (BD) with sparsity constraints, $(x_0,y_0)\in\bbC^{m_1}\times \bbC^{m_2}$ ($\|x_0\|_0\leq s_1$, $\|y_0\|_0\leq s_2$, $x_0\neq 0$, $y_0\neq 0$) is identifiable up to scaling, for almost all $D\in\bbC^{n\times m_1}$ and $E\in\bbC^{n\times m_2}$, if $n\geq 2s_1s_2$.
\end{theorem}
\begin{proof}
Fix index sets $J_0,J\subset\{1,2,\cdots,m_1\}$, for which $|J_0|=|J|=s_1$ and $|J_0\bigcap J|=t_1$, and index sets $K_0,K\subset\{1,2,\cdots,m_2\}$, for which $|K_0|=|K|=s_2$ and $|K_0\bigcap K|=t_2$. Let 
\begin{align*}
&D_0 = D^{(:,J_0\bigcap J)}\in\bbC^{n\times t_1}, \quad D_1 = D^{(:,J_0\setminus J)}\in\bbC^{n\times (s_1-t_1)}, \quad D_2 = D^{(:,J\setminus J_0)}\in\bbC^{n\times (s_1-t_1)},\\
&E_0 = E^{(:,K_0\bigcap K)}\in\bbC^{n\times t_2}, \quad E_1 = E^{(:,K_0\setminus K)}\in\bbC^{n\times (s_2-t_2)}, \quad E_2 = E^{(:,K\setminus K_0)}\in\bbC^{n\times (s_2-t_2)}.
\end{align*}
By Lemma \ref{lem:fcr3}, if $n\geq 2s_1s_2$, then for almost all $D$ and $E$, the columns of $G_{D_0E_0}$, $G_{D_1E_0}$, $G_{D_2E_0}$, $G_{D_0E_1}$, $G_{D_1E_1}$, $G_{D_0E_2}$, $G_{D_2E_2}$ together form a linearly independent set. For every $(x_0,y_0)$ and $(x,y)$ such that the $s_1$-sparse $x_0$ and $x$ are and supported on $J_0$ and $J$ respectively, and the $s_2$-sparse $y_0$ and $y$ are supported on $K_0$ and $K$ respectively, if $\mathcal{G}_{DE}(x_0y_0^T)=\mathcal{G}_{DE}(xy^T)$, then
\begin{align*}
& G_{D_0E_0}v_{00}+G_{D_1E_0}v_{10}+G_{D_2E_0}v_{20}+G_{D_0E_1}v_{01}+G_{D_1E_1}v_{11}+G_{D_0E_2}v_{02}+G_{D_2E_2}v_{22}\\
=& G_{DE}v=\mathcal{G}_{DE}(x_0y_0^T)-\mathcal{G}_{DE}(xy^T)=0,
\end{align*}
where $v=\operatorname{vec}(x_0y_0^T-xy^T)$, $v_{00} = \operatorname{vec}(x_0^{(J_0\bigcap J)}y_0^{(K_0\bigcap K)T}-x^{(J_0\bigcap J)}y^{(K_0\bigcap K)T})$, $v_{10} = \operatorname{vec}(x_0^{(J_0\setminus J)}y_0^{(K_0\bigcap K)T})$, $v_{20} = \operatorname{vec}(-x^{(J\setminus J_0)}y^{(K_0\bigcap K)T})$, $v_{01} = \operatorname{vec}(x_0^{(J_0\bigcap J)}y_0^{(K_0\setminus K)T})$, $v_{11} = \operatorname{vec}(x_0^{(J_0\setminus J)}y_0^{(K_0\setminus K)T})$, $v_{02} = \operatorname{vec}(-x^{(J_0\bigcap J)}y^{(K\setminus K_0)T})$, $v_{22} = \operatorname{vec}(-x^{(J\setminus J_0)}y^{(K\setminus K_0)T})$. By linear independence, the vectors $v_{00},v_{10},v_{20},v_{01},v_{11},v_{02},v_{22}$ are all zero vectors, and so is $v$. Hence for almost all $D$ and $E$, and all pairs $(x_0,y_0)$ and $(x,y)$ such that $x_0$ and $x$ are supported on $J_0$ and $J$ respectively, and $y_0$ and $y$ are supported on $K_0$ and $K$ respectively, if $\mathcal{G}_{DE}(x_0y_0^T)=\mathcal{G}_{DE}(xy^T)$, then $x_0y_0^T=xy^T$. Note that the supports $J_0, J, K_0, K$ are arbitrary, and there is only a finite number (${m_1 \choose s_1}^2{m_2 \choose s_2}^2$) of combinations of supports. Therefore, for almost all $D$ and $E$, every pair $(x_0,y_0)$ ($\|x_0\|_0\leq s_1$, $\|y_0\|_0\leq s_2$, $x_0\neq 0$, $y_0\neq 0$) is identifiable up to scaling.
\end{proof}

Due to symmetry, we can derive another sufficient condition for the scenario where $u=Dx$ resides in a $m_1$-dimensional subspace spanned by the columns of $D$, and $v=Ey$ is $s_2$-sparse over $E$.

For generic bases or frames, the above sample complexities $n\geq m_1m_2$, $n\geq 2s_1m_2$ or $n\geq 2s_1s_2$ are sufficient. These sampling complexities are not optimal, since they are in terms of the number of nonzero entries in $x_0y_0^T$, instead of the number of degrees of freedom in $x_0$ and $y_0$. For example, in the scenario with subspace constraints, Theorem \ref{thm:gebf1} requires $n\geq m_1m_2$ samples, versus the number of degrees of freedom, which is $m_1+m_2-1$. However, these results hold with essentially no assumptions on $D$ or $E$. They are the first algebraic sample complexities for blind deconvolution.

\section{Blind Deconvolution with a Sub-band Structured Basis} \label{sec:sbsb}
In this section, we consider the BD problem where the filter resides in a subspace spanned by a sub-band structured basis. For this setup, using the general framework for bilinear inverse problems we introduced recently in \cite{Li2015}, and Proposition \ref{pro:ibd} above, we derive much stronger, essentially optimal sample complexity results.

\begin{definition}\label{def:subband}
Let $\widetilde{E}=F_nE$, $E\in\bbC^{n\times m_2}$, and let $J_k$ denote the support of $\widetilde{E}^{(:,k)}$ ($1\leq k\leq m_2$). 
If
\[
\widehat{J}_k = J_k\setminus\left(\bigcup\limits_{k'\neq k}J_{k'}\right) \neq \emptyset \quad 
\text{for}\quad 1\leq k\leq m_2,
\]
then we say $E$ forms a sub-band structured basis. The nonempty index set $\widehat{J}_k$ and its cardinality $\ell_k\eqdef |\widehat{J}_k|$ are called the passband and the bandwidth of $E^{(:,k)}$, respectively.
\end{definition}

Like filters in a filter bank, the basis vectors of a sub-band structured basis are supported on different sub-bands in the Fourier domain (Figure \ref{fig:sbsb_a}). By Definition \ref{def:subband}, the sub-bands may overlap partially. For each sub-band, its passband consists of the frequency components (which need not be contiguous) that are not present in any other sub-band. For example, in acoustic signal processing or communications, an equalizer that adjusts the relative gains $y^{(k)}$ of different frequency components can be considered as a filter $v=Ey$ that resides in a subspace with a sub-band structured basis. See Figure \ref{fig:sbsb_b} for the DFTs of three different equalizers, and Figure \ref{fig:filterbank} for the filter bank implementation of an equalizer.
\begin{figure}[htbp]%
\centering
\subfigure[]{\input{sbsb_a.tex}\label{fig:sbsb_a}}
\subfigure[]{\input{sbsb_b.tex}\label{fig:sbsb_b}}
\caption{A sub-band structured basis. (a) DFTs of basis vectors. (b) Examples of frequency responses of filters in the span of the sub-band structured basis.}%
\label{fig:sbsb}%

\vspace{0.3in}
\includegraphics[width=0.5\columnwidth]{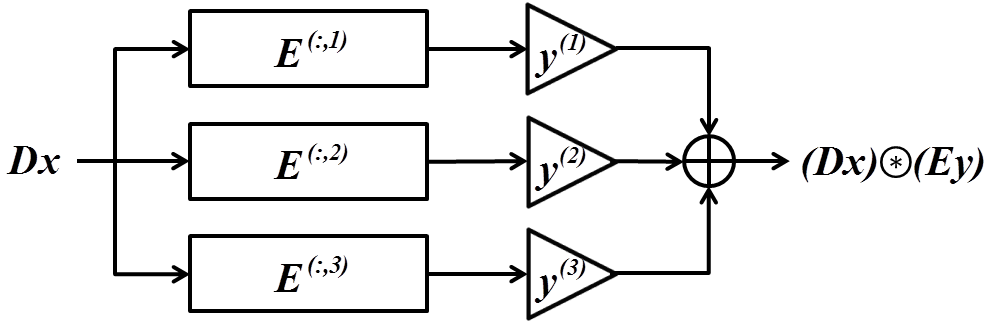}
\caption{Filter bank implementation of an equalizer.}%
\label{fig:filterbank}%
\end{figure}

Next, we address the identifiability of the blind deconvolution problem where the filter resides in a subspace with a sub-band structured basis, and the signal resides in another subspace, or is sparse over some given dictionary.
For example, consider the following blind deconvolution problem in channel encoding. An unknown source string $x$ is encoded by a given tall-and-skinny matrix $D$ and then transmitted through a channel whose gains in different sub-bands are unknown. Then the encoded string $Dx$ resides in a subspace spanned by $D$, and the channel resides in a subspace with a sub-band structured basis. Simultaneous recovery of the channel and the encoded string from measurements of the channel output corresponds to blind deconvolution with a sub-band structured basis.
Another example is the channel identification problem where the acoustic channel can be modeled as the serial concatentation of an equalizer and a multipath channel. The equalizer has known sub-bands but unknown gains. The multipath channel can be modeled as a sparse filter. Then the simultaneous recovery of the sparse multipath channel and the equalizer from the given input and measured output of the channel corresponds to blind deconvolution with a sub-band structured basis.

We consider first the case of subspace constraints, with one of the bases having a sub-band structure. 
\begin{theorem}\label{thm:sbsb}
In (BD) with subspace constraints, suppose $E$ forms a sub-band structured basis, $x_0\in\bbC^{m_1}$ is nonzero and $y_0\in\bbC^{m_2}$ is non-vanishing. If the sum of all the bandwidths $\sum_{k=1}^{m_2}\ell_k \geq m_1+m_2-1$,  then for almost all $D\in\bbC^{n\times m_1}$, the pair $(x_0,y_0)\in\Omega_\calX\times \Omega_\calY$ is identifiable up to scaling.
\end{theorem}

\begin{proof}
Let $\tilde{D} = F_nD$, $\tilde{E}=F_nE$. By the sub-band structure assumption, $\tilde{E}$ has full column rank. For nonzero $x_0$ and for almost all $D$, all the entries of $\widetilde{D}x_0$ are nonzero and the matrix $\diag(\widetilde{D}x_0)\widetilde{E}$ has full column rank. If there exists $y\in\Omega_\calY$ such that $(Dx_0)\circledast(Ey)=(Dx_0)\circledast(Ey_0)$, then
\[
\diag(\widetilde{D}x_0)\widetilde{E}y = (\widetilde{D}x_0)\odot(\widetilde{E}y)=(\widetilde{D}x_0)\odot(\widetilde{E}y_0)=\diag(\widetilde{D}x_0)\widetilde{E}y_0.
\]
It follows that $y=y_0$. By Proposition \ref{pro:ibd}, to complete the proof, we only need to show that if there exists $(x,y)\in\Omega_\calX\times \Omega_\calY$ such that $(Dx)\circledast(Ey)=(Dx_0)\circledast(Ey_0)$ then $x=\sigma x_0$ for some nonzero $\sigma$.

If there exists $(x,y)\in\Omega_\calX\times \Omega_\calY$ such that $(Dx)\circledast(Ey)=(Dx_0)\circledast(Ey_0)$, we have
\[
\diag(\widetilde{E}y)\widetilde{D}x = \diag(\widetilde{E}y_0)\widetilde{D}x_0.
\]
Considering the passband $\widehat{J}_k$, we have
\[
\diag(\widetilde{E}^{(\widehat{J}_k,:)}y)\widetilde{D}^{(\widehat{J}_k,:)}x = \diag(\widetilde{E}^{(\widehat{J}_k,:)}y_0)\widetilde{D}^{(\widehat{J}_k,:)}x_0,
\]
or equivalently
\begin{align*}
\diag(\widetilde{E}^{(\widehat{J}_k,k)})\widetilde{D}^{(\widehat{J}_k,:)}xy^{(k)} &= \diag(\widetilde{E}^{(\widehat{J}_k,k)})\widetilde{D}^{(\widehat{J}_k,:)}x_0y_0^{(k)}, \\
\widetilde{D}^{(\widehat{J}_k,:)}xy^{(k)} &= \widetilde{D}^{(\widehat{J}_k,:)}x_0y_0^{(k)}.
\end{align*}
By assumption, $y_0^{(k)}\neq 0$. For almost all $D$, $\widetilde{D}^{(\widehat{J}_k,:)}x_0\neq 0$. Hence $y^{(k)}\neq 0$, $x\neq 0$. It follows that
\[
\widetilde{D}^{(\widehat{J}_k,:)}(x-\frac{y_0^{(k)}}{y^{(k)}}x_0) = 0,
\]
Hence
\begin{equation}
x \in \calN(\widetilde{D}^{(\widehat{J}_k,:)})+\spanof(x_0).\label{eq:1}
\end{equation}
Let $x_0^\perp$ denote the orthogonal complement of $\spanof(x_0)$. Then
\begin{align*}
&P_{x_0^\perp} x \in x_0^\perp, \\
&P_{x_0^\perp} x = x-P_{\spanof(x_0)}x \in \calN(\widetilde{D}^{(\widehat{J}_k,:)})+\spanof(x_0).
\end{align*}
Hence
\begin{equation}
P_{x_0^\perp} x \in x_0^\perp \bigcap \left(\calN(\widetilde{D}^{(\widehat{J}_k,:)})+\spanof(x_0)\right) = x_0^\perp \bigcap \left(\calR(\widetilde{D}^{(\widehat{J}_k,:)*})\bigcap x_0^\perp \right)^\perp, \qquad \forall k\in\{1,2,\cdots,m_2\},\label{eq:2}
\end{equation}
where $(\cdot)^*$ denotes the conjugate transpose.
The equation holds due to the fact that, for linear vector spaces $\calV_1$ and $\calV_2$, $\calV_1+\calV_2=(\calV_1^\perp\bigcap \calV_2^\perp)^\perp$.

Now, for almost all $D$, $\calR(\widetilde{D}^{(\widehat{J}_k,:)*})$ is a generic $\ell_k$-dimensional subspace of $\bbC^{m_1}$, and $\calR(\widetilde{D}^{(\widehat{J}_k,:)*})\not\subset x_0^\perp$. Hence there exists a generic $(\ell_k-1)$-dimensional subspace $\calV_k\subset x_0^\perp$ such that
\begin{align*}
& \calR(\widetilde{D}^{(\widehat{J}_k,:)*}) = \calV_k \oplus \spanof(P_{\calR(\widetilde{D}^{(\widehat{J}_k,:)*})}x_0),\\
& \calR(\widetilde{D}^{(\widehat{J}_k,:)*})\bigcap x_0^\perp = \calV_k.
\end{align*}
Therefore, \eqref{eq:2} is equivalent to
\[
P_{x_0^\perp} x \in x_0^\perp \bigcap \calV_1^\perp \bigcap \calV_2^\perp \bigcap \cdots \bigcap \calV_{m_2}^\perp = (\spanof(x_0)+\sum_{k=1}^{m_2}\calV_k)^\perp,
\]
where $\calV_1,\calV_2,\cdots,\calV_{m_2}$ are generic subspaces of $x_0^\perp$, the dimensions of which are $\ell_1-1,\ell_2-1,\cdots,\ell_{m_2}-1$. For any such generic subspaces of $x_0^\perp$, if $\sum_{k=1}^{m_2}\ell_k \geq m_1+m_2-1$, i.e., $\sum_{k=1}^{m_2}(\ell_k-1) \geq m_1-1$, then,
\[
\sum_{k=1}^{m_2}\calV_k = x_0^\perp.
\]
Hence
\begin{align*}
\spanof(x_0)+\sum_{k=1}^{m_2}\calV_k = \bbC^{m_1},\\
P_{x_0^\perp} x \in(\spanof(x_0)+\sum_{k=1}^{m_2}\calV_k)^\perp = \{0\}.
\end{align*}
Therefore, $P_{x_0^\perp} x = 0$, or $x\in\spanof(x_0)$. We have shown that $x\neq 0$, hence there exists a nonzero $\sigma\in\bbC$ such that $x=\sigma x_0$. The proof is complete.
\end{proof}

We turn next to the case of blind deconvolution with mixed constraints, where the signal lives in a subspace spanned by a sub-band structured basis, and the filter is sparse.
\begin{theorem}\label{thm:sbsbmix}
In (BD) with mixed constraints, suppose $E$ forms a sub-band structured basis, $x_0\in\bbC^{m_1}$ satisfies that $\|x_0\|_0\leq s_1$ and $x_0\neq 0$, and $y_0\in\bbC^{m_2}$ is non-vanishing. If the sum of all the bandwidths $\sum_{k=1}^{m_2}\ell_k \geq 2s_1+m_2-1$, then for almost all $D\in\bbC^{n\times m_1}$, the pair $(x_0,y_0)\in\Omega_\calX\times \Omega_\calY$ is identifiable up to scaling.
\end{theorem}

\begin{proof}
The proof is very similar to that of Theorem \ref{thm:sbsb}. 
For nonzero $x_0$ and almost all $D$, if there exists $y\in\Omega_\calY$ such that $(Dx_0)\circledast(Ey)=(Dx_0)\circledast(Ey_0)$, then $y=y_0$. By Proposition \ref{pro:ibd}, to complete the proof, we only need to show that if there exists $(x,y)\in\Omega_\calX\times \Omega_\calY$ such that $\|x\|_0\leq s_1$ and $(Dx)\circledast(Ey)=(Dx_0)\circledast(Ey_0)$, then $x=\sigma x_0$ for some nonzero $\sigma$.

Denote the support of $x_0$ by $K_0$, $|K_0|= s_1$. If there exists $(x,y)\in\Omega_\calX\times \Omega_\calY$ such that $x$ is supported on $K$, $|K|= s_1$, and $(Dx)\circledast(Ey)=(Dx_0)\circledast(Ey_0)$, then
\[
\diag(\widetilde{E}y)\widetilde{D}^{(:,K_0\bigcup K)}x^{(K_0\bigcup K)} = \diag(\widetilde{E}y_0)\widetilde{D}^{(:,K_0\bigcup K)}x_0^{(K_0\bigcup K)}.
\]
In this case, \eqref{eq:1} and \eqref{eq:2} in the proof of Theorem \ref{thm:sbsb} become
\begin{align*}
& x^{(K_0\bigcup K)} \in \calN(\widetilde{D}^{(\widehat{J}_k,K_0\bigcup K)})+\spanof(x_0^{(K_0\bigcup K)}),\\
& P_{x_0^{(K_0\bigcup K)\perp}} x^{(K_0\bigcup K)} \in  x_0^{(K_0\bigcup K)\perp} \bigcap \left(\calR(\widetilde{D}^{(\widehat{J}_k,K_0\bigcup K)*})\bigcap x_0^{(K_0\bigcup K)\perp} \right)^\perp, \qquad \forall k\in\{1,2,\cdots,m_2\}.
\end{align*}
Since $|K_0|=|K|= s_1$, we have $|K_0\bigcup K|\leq 2s_1$. If $\sum_{k=1}^{m_2}\ell_k \geq 2s_1+m_2-1$, then by an argument analogous to that in the proof of Theorem \ref{thm:sbsb}, we have that for almost all $D$, $P_{x_0^{(K_0\bigcup K)\perp}} x^{(K_0\bigcup K)}$ must be $0$. Therefore, there exists a nonzero $\sigma\in\bbC$ such that $x=\sigma x_0$. 

We complete the proof by enumerating all supports $K$ of cardinality $s_1$. Since there is only a finite number (${m_1\choose s_1}$) of such supports, for almost all $D$, if there exists $(x,y)$ such that $x$ is $s_1$-sparse and $(Dx)\circledast(Ey)=(Dx_0)\circledast(Ey_0)$, then $x=\sigma x_0$ for some nonzero $\sigma$. 
\end{proof}

How do the sufficient conditions of Theorems \ref{thm:sbsb} and \ref{thm:sbsbmix} compare to the minimal required sample complexities? We address this question for the following scenario.
Suppose that the supports $J_k$ ($1\leq k\leq m_2$) form a partition of the frequency range, i.e.,
\begin{align*}
& J_{k_1}\bigcap J_{k_2} = \emptyset \quad\text{for all $k_1$ and $k_2$ such that $k_1\neq k_2$},\\
& \bigcup_{1\leq k\leq m_2}J_k = \{1,2,\cdots,n\}.
\end{align*}
In this case the passbands are $\widehat{J}_k=J_k$ and $n=\sum_{k=1}^{m_2}\ell_k$. For example, this scenario applies when the filter bank is an array of ideal bandpass filters whose passbands partition the DFT frequency range (See Figure \ref{fig:ideal}). Consider first (BD) with subspace constraints. Under the above scenario, the sufficient condition in Theorem \ref{thm:sbsb} implies $n\geq m_1+m_2-1$. Next, we show that this sample complexity is also necessary.

\begin{figure}
\centering
\subfigure[]{
%
%
\begin{tikzpicture}[scale=0.7]

\begin{axis}[%
width=3in,
height=0.7in,
scale only axis,
xmin=0,
xmax=28,
ymin=0,
ymax=1.2,
hide axis,
name=plot2,
axis x line*=bottom,
axis y line*=left
]
\addplot [color=white,solid,forget plot]
  table[row sep=crcr]{%
25	0\\
28	0\\
};
\addplot [color=white,solid,forget plot]
  table[row sep=crcr]{%
28	0\\
28	1.2\\
};
\addplot [color=blue,solid,line width=1.0pt,forget plot]
  table[row sep=crcr]{%
0	0\\
25	0\\
};
\addplot [color=blue,only marks,mark=o,mark options={solid},forget plot]
  table[row sep=crcr]{%
1	0\\
2	0\\
3	0\\
4	0\\
5	0\\
6	0\\
7	0\\
8	0\\
9	1\\
10	1\\
11	1\\
12	1\\
13	1\\
14	1\\
15	1\\
16	1\\
17	0\\
18	0\\
19	0\\
20	0\\
21	0\\
22	0\\
23	0\\
24	0\\
};
\addplot [color=blue,solid,line width=2.0pt,forget plot]
  table[row sep=crcr]{%
1	0\\
1	0\\
};
\addplot [color=blue,solid,line width=2.0pt,forget plot]
  table[row sep=crcr]{%
2	0\\
2	0\\
};
\addplot [color=blue,solid,line width=2.0pt,forget plot]
  table[row sep=crcr]{%
3	0\\
3	0\\
};
\addplot [color=blue,solid,line width=2.0pt,forget plot]
  table[row sep=crcr]{%
4	0\\
4	0\\
};
\addplot [color=blue,solid,line width=2.0pt,forget plot]
  table[row sep=crcr]{%
5	0\\
5	0\\
};
\addplot [color=blue,solid,line width=2.0pt,forget plot]
  table[row sep=crcr]{%
6	0\\
6	0\\
};
\addplot [color=blue,solid,line width=2.0pt,forget plot]
  table[row sep=crcr]{%
7	0\\
7	0\\
};
\addplot [color=blue,solid,line width=2.0pt,forget plot]
  table[row sep=crcr]{%
8	0\\
8	0\\
};
\addplot [color=blue,solid,line width=2.0pt,forget plot]
  table[row sep=crcr]{%
9	0\\
9	1\\
};
\addplot [color=blue,solid,line width=2.0pt,forget plot]
  table[row sep=crcr]{%
10	0\\
10	1\\
};
\addplot [color=blue,solid,line width=2.0pt,forget plot]
  table[row sep=crcr]{%
11	0\\
11	1\\
};
\addplot [color=blue,solid,line width=2.0pt,forget plot]
  table[row sep=crcr]{%
12	0\\
12	1\\
};
\addplot [color=blue,solid,line width=2.0pt,forget plot]
  table[row sep=crcr]{%
13	0\\
13	1\\
};
\addplot [color=blue,solid,line width=2.0pt,forget plot]
  table[row sep=crcr]{%
14	0\\
14	1\\
};
\addplot [color=blue,solid,line width=2.0pt,forget plot]
  table[row sep=crcr]{%
15	0\\
15	1\\
};
\addplot [color=blue,solid,line width=2.0pt,forget plot]
  table[row sep=crcr]{%
16	0\\
16	1\\
};
\addplot [color=blue,solid,line width=2.0pt,forget plot]
  table[row sep=crcr]{%
17	0\\
17	0\\
};
\addplot [color=blue,solid,line width=2.0pt,forget plot]
  table[row sep=crcr]{%
18	0\\
18	0\\
};
\addplot [color=blue,solid,line width=2.0pt,forget plot]
  table[row sep=crcr]{%
19	0\\
19	0\\
};
\addplot [color=blue,solid,line width=2.0pt,forget plot]
  table[row sep=crcr]{%
20	0\\
20	0\\
};
\addplot [color=blue,solid,line width=2.0pt,forget plot]
  table[row sep=crcr]{%
21	0\\
21	0\\
};
\addplot [color=blue,solid,line width=2.0pt,forget plot]
  table[row sep=crcr]{%
22	0\\
22	0\\
};
\addplot [color=blue,solid,line width=2.0pt,forget plot]
  table[row sep=crcr]{%
23	0\\
23	0\\
};
\addplot [color=blue,solid,line width=2.0pt,forget plot]
  table[row sep=crcr]{%
24	0\\
24	0\\
};
\end{axis}

\begin{axis}[%
width=3in,
height=0.7in,
scale only axis,
xmin=0,
xmax=28,
ymin=0,
ymax=1.2,
hide axis,
name=plot1,
at=(plot2.above north west),
anchor=below south west,
axis x line*=bottom,
axis y line*=left
]
\addplot [color=white,solid,forget plot]
  table[row sep=crcr]{%
25	0\\
28	0\\
};
\addplot [color=white,solid,forget plot]
  table[row sep=crcr]{%
28	0\\
28	1.2\\
};
\addplot [color=blue,solid,line width=1.0pt,forget plot]
  table[row sep=crcr]{%
0	0\\
25	0\\
};
\addplot [color=blue,only marks,mark=o,mark options={solid},forget plot]
  table[row sep=crcr]{%
1	1\\
2	1\\
3	1\\
4	1\\
5	1\\
6	1\\
7	1\\
8	1\\
9	0\\
10	0\\
11	0\\
12	0\\
13	0\\
14	0\\
15	0\\
16	0\\
17	0\\
18	0\\
19	0\\
20	0\\
21	0\\
22	0\\
23	0\\
24	0\\
};
\addplot [color=blue,solid,line width=2.0pt,forget plot]
  table[row sep=crcr]{%
1	0\\
1	1\\
};
\addplot [color=blue,solid,line width=2.0pt,forget plot]
  table[row sep=crcr]{%
2	0\\
2	1\\
};
\addplot [color=blue,solid,line width=2.0pt,forget plot]
  table[row sep=crcr]{%
3	0\\
3	1\\
};
\addplot [color=blue,solid,line width=2.0pt,forget plot]
  table[row sep=crcr]{%
4	0\\
4	1\\
};
\addplot [color=blue,solid,line width=2.0pt,forget plot]
  table[row sep=crcr]{%
5	0\\
5	1\\
};
\addplot [color=blue,solid,line width=2.0pt,forget plot]
  table[row sep=crcr]{%
6	0\\
6	1\\
};
\addplot [color=blue,solid,line width=2.0pt,forget plot]
  table[row sep=crcr]{%
7	0\\
7	1\\
};
\addplot [color=blue,solid,line width=2.0pt,forget plot]
  table[row sep=crcr]{%
8	0\\
8	1\\
};
\addplot [color=blue,solid,line width=2.0pt,forget plot]
  table[row sep=crcr]{%
9	0\\
9	0\\
};
\addplot [color=blue,solid,line width=2.0pt,forget plot]
  table[row sep=crcr]{%
10	0\\
10	0\\
};
\addplot [color=blue,solid,line width=2.0pt,forget plot]
  table[row sep=crcr]{%
11	0\\
11	0\\
};
\addplot [color=blue,solid,line width=2.0pt,forget plot]
  table[row sep=crcr]{%
12	0\\
12	0\\
};
\addplot [color=blue,solid,line width=2.0pt,forget plot]
  table[row sep=crcr]{%
13	0\\
13	0\\
};
\addplot [color=blue,solid,line width=2.0pt,forget plot]
  table[row sep=crcr]{%
14	0\\
14	0\\
};
\addplot [color=blue,solid,line width=2.0pt,forget plot]
  table[row sep=crcr]{%
15	0\\
15	0\\
};
\addplot [color=blue,solid,line width=2.0pt,forget plot]
  table[row sep=crcr]{%
16	0\\
16	0\\
};
\addplot [color=blue,solid,line width=2.0pt,forget plot]
  table[row sep=crcr]{%
17	0\\
17	0\\
};
\addplot [color=blue,solid,line width=2.0pt,forget plot]
  table[row sep=crcr]{%
18	0\\
18	0\\
};
\addplot [color=blue,solid,line width=2.0pt,forget plot]
  table[row sep=crcr]{%
19	0\\
19	0\\
};
\addplot [color=blue,solid,line width=2.0pt,forget plot]
  table[row sep=crcr]{%
20	0\\
20	0\\
};
\addplot [color=blue,solid,line width=2.0pt,forget plot]
  table[row sep=crcr]{%
21	0\\
21	0\\
};
\addplot [color=blue,solid,line width=2.0pt,forget plot]
  table[row sep=crcr]{%
22	0\\
22	0\\
};
\addplot [color=blue,solid,line width=2.0pt,forget plot]
  table[row sep=crcr]{%
23	0\\
23	0\\
};
\addplot [color=blue,solid,line width=2.0pt,forget plot]
  table[row sep=crcr]{%
24	0\\
24	0\\
};
\end{axis}

\begin{axis}[%
width=3in,
height=0.7in,
scale only axis,
xmin=0,
xmax=28,
ymin=0,
ymax=1.2,
hide axis,
at=(plot2.below south west),
anchor=above north west,
axis x line*=bottom,
axis y line*=left
]
\addplot [color=white,solid,forget plot]
  table[row sep=crcr]{%
25	0\\
28	0\\
};
\addplot [color=white,solid,forget plot]
  table[row sep=crcr]{%
28	0\\
28	1.2\\
};
\addplot [color=blue,solid,line width=1.0pt,forget plot]
  table[row sep=crcr]{%
0	0\\
25	0\\
};
\addplot [color=blue,only marks,mark=o,mark options={solid},forget plot]
  table[row sep=crcr]{%
1	0\\
2	0\\
3	0\\
4	0\\
5	0\\
6	0\\
7	0\\
8	0\\
9	0\\
10	0\\
11	0\\
12	0\\
13	0\\
14	0\\
15	0\\
16	0\\
17	1\\
18	1\\
19	1\\
20	1\\
21	1\\
22	1\\
23	1\\
24	1\\
};
\addplot [color=blue,solid,line width=2.0pt,forget plot]
  table[row sep=crcr]{%
1	0\\
1	0\\
};
\addplot [color=blue,solid,line width=2.0pt,forget plot]
  table[row sep=crcr]{%
2	0\\
2	0\\
};
\addplot [color=blue,solid,line width=2.0pt,forget plot]
  table[row sep=crcr]{%
3	0\\
3	0\\
};
\addplot [color=blue,solid,line width=2.0pt,forget plot]
  table[row sep=crcr]{%
4	0\\
4	0\\
};
\addplot [color=blue,solid,line width=2.0pt,forget plot]
  table[row sep=crcr]{%
5	0\\
5	0\\
};
\addplot [color=blue,solid,line width=2.0pt,forget plot]
  table[row sep=crcr]{%
6	0\\
6	0\\
};
\addplot [color=blue,solid,line width=2.0pt,forget plot]
  table[row sep=crcr]{%
7	0\\
7	0\\
};
\addplot [color=blue,solid,line width=2.0pt,forget plot]
  table[row sep=crcr]{%
8	0\\
8	0\\
};
\addplot [color=blue,solid,line width=2.0pt,forget plot]
  table[row sep=crcr]{%
9	0\\
9	0\\
};
\addplot [color=blue,solid,line width=2.0pt,forget plot]
  table[row sep=crcr]{%
10	0\\
10	0\\
};
\addplot [color=blue,solid,line width=2.0pt,forget plot]
  table[row sep=crcr]{%
11	0\\
11	0\\
};
\addplot [color=blue,solid,line width=2.0pt,forget plot]
  table[row sep=crcr]{%
12	0\\
12	0\\
};
\addplot [color=blue,solid,line width=2.0pt,forget plot]
  table[row sep=crcr]{%
13	0\\
13	0\\
};
\addplot [color=blue,solid,line width=2.0pt,forget plot]
  table[row sep=crcr]{%
14	0\\
14	0\\
};
\addplot [color=blue,solid,line width=2.0pt,forget plot]
  table[row sep=crcr]{%
15	0\\
15	0\\
};
\addplot [color=blue,solid,line width=2.0pt,forget plot]
  table[row sep=crcr]{%
16	0\\
16	0\\
};
\addplot [color=blue,solid,line width=2.0pt,forget plot]
  table[row sep=crcr]{%
17	0\\
17	1\\
};
\addplot [color=blue,solid,line width=2.0pt,forget plot]
  table[row sep=crcr]{%
18	0\\
18	1\\
};
\addplot [color=blue,solid,line width=2.0pt,forget plot]
  table[row sep=crcr]{%
19	0\\
19	1\\
};
\addplot [color=blue,solid,line width=2.0pt,forget plot]
  table[row sep=crcr]{%
20	0\\
20	1\\
};
\addplot [color=blue,solid,line width=2.0pt,forget plot]
  table[row sep=crcr]{%
21	0\\
21	1\\
};
\addplot [color=blue,solid,line width=2.0pt,forget plot]
  table[row sep=crcr]{%
22	0\\
22	1\\
};
\addplot [color=blue,solid,line width=2.0pt,forget plot]
  table[row sep=crcr]{%
23	0\\
23	1\\
};
\addplot [color=blue,solid,line width=2.0pt,forget plot]
  table[row sep=crcr]{%
24	0\\
24	1\\
};
\end{axis}
\end{tikzpicture}%
\label{fig:ideal_a}}
\subfigure[]{
%
%
\begin{tikzpicture}[scale=0.7]

\begin{axis}[%
width=3in,
height=0.7in,
scale only axis,
xmin=0,
xmax=28,
ymin=0,
ymax=1.2,
hide axis,
name=plot4,
axis x line*=bottom,
axis y line*=left
]
\addplot [color=white,solid,forget plot]
  table[row sep=crcr]{%
25	0\\
28	0\\
};
\addplot [color=white,solid,forget plot]
  table[row sep=crcr]{%
28	0\\
28	1.2\\
};
\addplot [color=blue,solid,line width=1.0pt,forget plot]
  table[row sep=crcr]{%
0	0\\
25	0\\
};
\addplot [color=blue,only marks,mark=o,mark options={solid},forget plot]
  table[row sep=crcr]{%
1	0.5\\
2	0.5\\
3	0.5\\
4	0.5\\
5	0.5\\
6	0.5\\
7	0.5\\
8	0.5\\
9	1\\
10	1\\
11	1\\
12	1\\
13	1\\
14	1\\
15	1\\
16	1\\
17	0.5\\
18	0.5\\
19	0.5\\
20	0.5\\
21	0.5\\
22	0.5\\
23	0.5\\
24	0.5\\
};
\addplot [color=blue,solid,line width=2.0pt,forget plot]
  table[row sep=crcr]{%
1	0\\
1	0.5\\
};
\addplot [color=blue,solid,line width=2.0pt,forget plot]
  table[row sep=crcr]{%
2	0\\
2	0.5\\
};
\addplot [color=blue,solid,line width=2.0pt,forget plot]
  table[row sep=crcr]{%
3	0\\
3	0.5\\
};
\addplot [color=blue,solid,line width=2.0pt,forget plot]
  table[row sep=crcr]{%
4	0\\
4	0.5\\
};
\addplot [color=blue,solid,line width=2.0pt,forget plot]
  table[row sep=crcr]{%
5	0\\
5	0.5\\
};
\addplot [color=blue,solid,line width=2.0pt,forget plot]
  table[row sep=crcr]{%
6	0\\
6	0.5\\
};
\addplot [color=blue,solid,line width=2.0pt,forget plot]
  table[row sep=crcr]{%
7	0\\
7	0.5\\
};
\addplot [color=blue,solid,line width=2.0pt,forget plot]
  table[row sep=crcr]{%
8	0\\
8	0.5\\
};
\addplot [color=blue,solid,line width=2.0pt,forget plot]
  table[row sep=crcr]{%
9	0\\
9	1\\
};
\addplot [color=blue,solid,line width=2.0pt,forget plot]
  table[row sep=crcr]{%
10	0\\
10	1\\
};
\addplot [color=blue,solid,line width=2.0pt,forget plot]
  table[row sep=crcr]{%
11	0\\
11	1\\
};
\addplot [color=blue,solid,line width=2.0pt,forget plot]
  table[row sep=crcr]{%
12	0\\
12	1\\
};
\addplot [color=blue,solid,line width=2.0pt,forget plot]
  table[row sep=crcr]{%
13	0\\
13	1\\
};
\addplot [color=blue,solid,line width=2.0pt,forget plot]
  table[row sep=crcr]{%
14	0\\
14	1\\
};
\addplot [color=blue,solid,line width=2.0pt,forget plot]
  table[row sep=crcr]{%
15	0\\
15	1\\
};
\addplot [color=blue,solid,line width=2.0pt,forget plot]
  table[row sep=crcr]{%
16	0\\
16	1\\
};
\addplot [color=blue,solid,line width=2.0pt,forget plot]
  table[row sep=crcr]{%
17	0\\
17	0.5\\
};
\addplot [color=blue,solid,line width=2.0pt,forget plot]
  table[row sep=crcr]{%
18	0\\
18	0.5\\
};
\addplot [color=blue,solid,line width=2.0pt,forget plot]
  table[row sep=crcr]{%
19	0\\
19	0.5\\
};
\addplot [color=blue,solid,line width=2.0pt,forget plot]
  table[row sep=crcr]{%
20	0\\
20	0.5\\
};
\addplot [color=blue,solid,line width=2.0pt,forget plot]
  table[row sep=crcr]{%
21	0\\
21	0.5\\
};
\addplot [color=blue,solid,line width=2.0pt,forget plot]
  table[row sep=crcr]{%
22	0\\
22	0.5\\
};
\addplot [color=blue,solid,line width=2.0pt,forget plot]
  table[row sep=crcr]{%
23	0\\
23	0.5\\
};
\addplot [color=blue,solid,line width=2.0pt,forget plot]
  table[row sep=crcr]{%
24	0\\
24	0.5\\
};
\end{axis}

\begin{axis}[%
width=3in,
height=0.7in,
scale only axis,
xmin=0,
xmax=28,
ymin=0,
ymax=1.2,
hide axis,
name=plot5,
at=(plot4.below south west),
anchor=above north west,
axis x line*=bottom,
axis y line*=left
]
\addplot [color=white,solid,forget plot]
  table[row sep=crcr]{%
25	0\\
28	0\\
};
\addplot [color=white,solid,forget plot]
  table[row sep=crcr]{%
28	0\\
28	1.2\\
};
\addplot [color=blue,solid,line width=1.0pt,forget plot]
  table[row sep=crcr]{%
0	0\\
25	0\\
};
\addplot [color=blue,only marks,mark=o,mark options={solid},forget plot]
  table[row sep=crcr]{%
1	1\\
2	1\\
3	1\\
4	1\\
5	1\\
6	1\\
7	1\\
8	1\\
9	0.7\\
10	0.7\\
11	0.7\\
12	0.7\\
13	0.7\\
14	0.7\\
15	0.7\\
16	0.7\\
17	0.3\\
18	0.3\\
19	0.3\\
20	0.3\\
21	0.3\\
22	0.3\\
23	0.3\\
24	0.3\\
};
\addplot [color=blue,solid,line width=2.0pt,forget plot]
  table[row sep=crcr]{%
1	0\\
1	1\\
};
\addplot [color=blue,solid,line width=2.0pt,forget plot]
  table[row sep=crcr]{%
2	0\\
2	1\\
};
\addplot [color=blue,solid,line width=2.0pt,forget plot]
  table[row sep=crcr]{%
3	0\\
3	1\\
};
\addplot [color=blue,solid,line width=2.0pt,forget plot]
  table[row sep=crcr]{%
4	0\\
4	1\\
};
\addplot [color=blue,solid,line width=2.0pt,forget plot]
  table[row sep=crcr]{%
5	0\\
5	1\\
};
\addplot [color=blue,solid,line width=2.0pt,forget plot]
  table[row sep=crcr]{%
6	0\\
6	1\\
};
\addplot [color=blue,solid,line width=2.0pt,forget plot]
  table[row sep=crcr]{%
7	0\\
7	1\\
};
\addplot [color=blue,solid,line width=2.0pt,forget plot]
  table[row sep=crcr]{%
8	0\\
8	1\\
};
\addplot [color=blue,solid,line width=2.0pt,forget plot]
  table[row sep=crcr]{%
9	0\\
9	0.7\\
};
\addplot [color=blue,solid,line width=2.0pt,forget plot]
  table[row sep=crcr]{%
10	0\\
10	0.7\\
};
\addplot [color=blue,solid,line width=2.0pt,forget plot]
  table[row sep=crcr]{%
11	0\\
11	0.7\\
};
\addplot [color=blue,solid,line width=2.0pt,forget plot]
  table[row sep=crcr]{%
12	0\\
12	0.7\\
};
\addplot [color=blue,solid,line width=2.0pt,forget plot]
  table[row sep=crcr]{%
13	0\\
13	0.7\\
};
\addplot [color=blue,solid,line width=2.0pt,forget plot]
  table[row sep=crcr]{%
14	0\\
14	0.7\\
};
\addplot [color=blue,solid,line width=2.0pt,forget plot]
  table[row sep=crcr]{%
15	0\\
15	0.7\\
};
\addplot [color=blue,solid,line width=2.0pt,forget plot]
  table[row sep=crcr]{%
16	0\\
16	0.7\\
};
\addplot [color=blue,solid,line width=2.0pt,forget plot]
  table[row sep=crcr]{%
17	0\\
17	0.3\\
};
\addplot [color=blue,solid,line width=2.0pt,forget plot]
  table[row sep=crcr]{%
18	0\\
18	0.3\\
};
\addplot [color=blue,solid,line width=2.0pt,forget plot]
  table[row sep=crcr]{%
19	0\\
19	0.3\\
};
\addplot [color=blue,solid,line width=2.0pt,forget plot]
  table[row sep=crcr]{%
20	0\\
20	0.3\\
};
\addplot [color=blue,solid,line width=2.0pt,forget plot]
  table[row sep=crcr]{%
21	0\\
21	0.3\\
};
\addplot [color=blue,solid,line width=2.0pt,forget plot]
  table[row sep=crcr]{%
22	0\\
22	0.3\\
};
\addplot [color=blue,solid,line width=2.0pt,forget plot]
  table[row sep=crcr]{%
23	0\\
23	0.3\\
};
\addplot [color=blue,solid,line width=2.0pt,forget plot]
  table[row sep=crcr]{%
24	0\\
24	0.3\\
};
\end{axis}

\begin{axis}[%
width=3in,
height=0.7in,
scale only axis,
xmin=0,
xmax=28,
ymin=0,
ymax=1.2,
hide axis,
name=plot6,
at=(plot5.below south west),
anchor=above north west,
axis x line*=bottom,
axis y line*=left
]
\addplot [color=white,solid,forget plot]
  table[row sep=crcr]{%
25	0\\
28	0\\
};
\addplot [color=white,solid,forget plot]
  table[row sep=crcr]{%
28	0\\
28	1.2\\
};
\addplot [color=blue,solid,line width=1.0pt,forget plot]
  table[row sep=crcr]{%
0	0\\
25	0\\
};
\addplot [color=blue,only marks,mark=o,mark options={solid},forget plot]
  table[row sep=crcr]{%
1	0.9\\
2	0.9\\
3	0.9\\
4	0.9\\
5	0.9\\
6	0.9\\
7	0.9\\
8	0.9\\
9	0.2\\
10	0.2\\
11	0.2\\
12	0.2\\
13	0.2\\
14	0.2\\
15	0.2\\
16	0.2\\
17	0.9\\
18	0.9\\
19	0.9\\
20	0.9\\
21	0.9\\
22	0.9\\
23	0.9\\
24	0.9\\
};
\addplot [color=blue,solid,line width=2.0pt,forget plot]
  table[row sep=crcr]{%
1	0\\
1	0.9\\
};
\addplot [color=blue,solid,line width=2.0pt,forget plot]
  table[row sep=crcr]{%
2	0\\
2	0.9\\
};
\addplot [color=blue,solid,line width=2.0pt,forget plot]
  table[row sep=crcr]{%
3	0\\
3	0.9\\
};
\addplot [color=blue,solid,line width=2.0pt,forget plot]
  table[row sep=crcr]{%
4	0\\
4	0.9\\
};
\addplot [color=blue,solid,line width=2.0pt,forget plot]
  table[row sep=crcr]{%
5	0\\
5	0.9\\
};
\addplot [color=blue,solid,line width=2.0pt,forget plot]
  table[row sep=crcr]{%
6	0\\
6	0.9\\
};
\addplot [color=blue,solid,line width=2.0pt,forget plot]
  table[row sep=crcr]{%
7	0\\
7	0.9\\
};
\addplot [color=blue,solid,line width=2.0pt,forget plot]
  table[row sep=crcr]{%
8	0\\
8	0.9\\
};
\addplot [color=blue,solid,line width=2.0pt,forget plot]
  table[row sep=crcr]{%
9	0\\
9	0.2\\
};
\addplot [color=blue,solid,line width=2.0pt,forget plot]
  table[row sep=crcr]{%
10	0\\
10	0.2\\
};
\addplot [color=blue,solid,line width=2.0pt,forget plot]
  table[row sep=crcr]{%
11	0\\
11	0.2\\
};
\addplot [color=blue,solid,line width=2.0pt,forget plot]
  table[row sep=crcr]{%
12	0\\
12	0.2\\
};
\addplot [color=blue,solid,line width=2.0pt,forget plot]
  table[row sep=crcr]{%
13	0\\
13	0.2\\
};
\addplot [color=blue,solid,line width=2.0pt,forget plot]
  table[row sep=crcr]{%
14	0\\
14	0.2\\
};
\addplot [color=blue,solid,line width=2.0pt,forget plot]
  table[row sep=crcr]{%
15	0\\
15	0.2\\
};
\addplot [color=blue,solid,line width=2.0pt,forget plot]
  table[row sep=crcr]{%
16	0\\
16	0.2\\
};
\addplot [color=blue,solid,line width=2.0pt,forget plot]
  table[row sep=crcr]{%
17	0\\
17	0.9\\
};
\addplot [color=blue,solid,line width=2.0pt,forget plot]
  table[row sep=crcr]{%
18	0\\
18	0.9\\
};
\addplot [color=blue,solid,line width=2.0pt,forget plot]
  table[row sep=crcr]{%
19	0\\
19	0.9\\
};
\addplot [color=blue,solid,line width=2.0pt,forget plot]
  table[row sep=crcr]{%
20	0\\
20	0.9\\
};
\addplot [color=blue,solid,line width=2.0pt,forget plot]
  table[row sep=crcr]{%
21	0\\
21	0.9\\
};
\addplot [color=blue,solid,line width=2.0pt,forget plot]
  table[row sep=crcr]{%
22	0\\
22	0.9\\
};
\addplot [color=blue,solid,line width=2.0pt,forget plot]
  table[row sep=crcr]{%
23	0\\
23	0.9\\
};
\addplot [color=blue,solid,line width=2.0pt,forget plot]
  table[row sep=crcr]{%
24	0\\
24	0.9\\
};
\end{axis}
\end{tikzpicture}%
\label{fig:ideal_b}}
\caption{A sub-band structured basis with supports $J_k$ that partition the DFT frequency range. (a) DFTs of basis vectors. (b) Examples of frequency responses of filters in the span of the basis.}%
\label{fig:ideal}%
\end{figure}
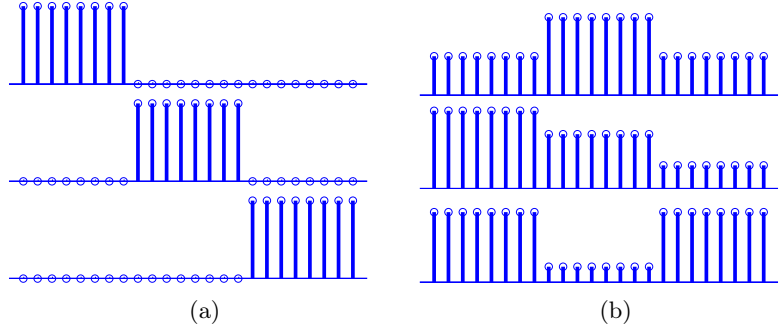

\begin{proposition}\label{pro:necessary1}
In (BD) with subspace constraints, suppose $E$ forms a sub-band structured basis, for which the supports $J_k$ ($1\leq k\leq m_2$) are disjoint and cover all the frequency components. If $(x_0,y_0)$ ($y_0$ is non-vanishing) is identifiable up to scaling, then $n\geq m_1+m_2-1$.
\end{proposition}

We turn next to (BD) with mixed constraints. Under the assumption that the passbands partition the DFT frequency range, the sufficient condition in Theorem \ref{thm:sbsbmix} implies $n\geq 2s_1+m_2-1$. Next, we show that this is almost necessary. 

\begin{corollary}\label{cor:necessary2}
In (BD) with mixed constraints, suppose $E$ forms a sub-band structured basis, for which the supports $J_k$ ($1\leq k\leq m_2$) are disjoint and cover all the frequency components. If $(x_0,y_0)$ ($x_0$ is $s_1$-sparse, $y_0$ is non-vanishing) is identifiable up to scaling, then $n\geq s_1+m_2-1$.
\end{corollary}

The sample complexities in the sufficient conditions match (exactly for (BD) with subspace constraints and almost for (BD) with mixed constraints) those in the necessary conditions, hence they are optimal. The sample complexities are also optimal in the sense that the number of degrees of freedom is roughly equal to the number of measurements. We give the proofs of Proposition \ref{pro:necessary1} and Corollary \ref{cor:necessary2} in Appendix \ref{app:necessary}.

\section{Conclusions} \label{sec:conclusions}
We studied the identifiability of blind deconvolution problems with subspace or sparsity constraints. We derived two algebraic conditions on blind deconvolution with subspace constraints. We first showed using the lifting framework that blind deconvolution from $n$ observations with generic bases of dimensions $m_1$ and $m_2$ is identifiable up to scaling given that $n\geq m_1m_2$. Then we applied the general framework in \cite{Li2015} to show that blind deconvolution with a sub-band structured basis is identifiable up to scaling given that $n\geq m_1+m_2-1$. The second result was shown to be tight. These results are also generalized to blind deconvolution with sparsity constraints or mixed constraints, with sparsity level(s) replacing the subspace dimension(s). The extra cost for the unknown support in the case of sparsity constraints is an extra factor of $2$ in the sample complexity.

We acknowledge that the results in Section \ref{sec:gebf} for generic bases may not be optimal. But they provide the first algebraic conditions for feasibility of blind deconvolution with subspace or sparsity priors. Furthermore, taking advantage of the interesting sub-band structure of some bases (such as filters in a filter bank implementation of equalizers), we achieved sample complexities that are essentially optimal. Our results are derived with generic bases or frames, which means they are violated on a set of Lebesgue measure zero.

An interesting question is, without the sub-band structure, whether or not it is possible to provide an algebraic analysis of blind deconvolution that achieves optimal sample complexities. Other ongoing research topics includes identifiability in blind deconvolution with specific bases or frames that arise in applications.

\appendix
\section{Proofs of Lemma \ref{lem:fcr1}, \ref{lem:fcr2} and \ref{lem:fcr3}}\label{app:fcr}
\begin{proof}[Proof of Lemma \ref{lem:fcr1}]
The entries of $G_{DE}$ are multivariate polynomials in the entries of $D$ and $E$, or to be more specific, quadratic forms in the entries of $D$ and $E$. By Lemma 1 from \cite{Harikumar1998}, the matrix $G_{DE}$ has full column rank for almost all $D$ and $E$ if it has full column rank for at least one choice of $D$ and $E$.

We complete the proof by by showing that $G_{DE}$ has full column rank for the following choice of $D$ and $E$. Let $D=F_n^{-1}\widetilde{D}$, with $\widetilde{D}\in\bbC^{n\times m_1}$ chosen such that all its submatrices have full rank. (For example, this will hold with probability $1$ for a random matrix with iid Gaussian entries.) Let $E=F_n^{-1}\widetilde{E}$, with $\widetilde{E}\in\bbC^{n\times m_2}$ chosen such that the first $m_1m_2$ rows are the kronecker product:
\[
\widetilde{E}^{(1:m_1m_2,:)} = I_{m_2}\otimes \mathbf{1}_{m_1,1}.
\]
Let $\widetilde{G}_{DE}=F_nG_{DE}$, then the submatrix containing the first $m_1m_2$ rows of $\widetilde{G}_{DE}/\sqrt{n}$ is
\[
\frac{1}{\sqrt{n}}\widetilde{G}_{DE}^{(1:m_1m_2,:)} =
\begin{bmatrix}
\widetilde{D}^{(1:m_1,:)} &  &  &  \\
 & \widetilde{D}^{(m_1+1:2m_1,:)} &  &  \\
 &  & \ddots &  \\
 &  &  & \widetilde{D}^{(m_1m_2-m_1+1:m_1m_2,:)}
\end{bmatrix}.
\]
By the assumption that all submatrices of $\widetilde{D}$ have full rank, it follows that $\widetilde{G}_{DE}^{(1:m_1m_2,:)}/{\sqrt{n}}$ has full column rank $m_1m_2$. Therefore, $G_{DE}$ has full column rank.
\end{proof}

\begin{proof}[Proof of Lemma \ref{lem:fcr2}]
Let $D=[D_0, D_1, D_2]$, then $G_{DE}$ is a permutation of the columns of $[G_{D_0E},G_{D_1E},G_{D_2E}]$. It is sufficient to prove that $G_{DE}$ has full column rank, which follows from Lemma \ref{lem:fcr1} because the number of columns in $D$ is $m_1=t_1+2\times(s_1-t_1)=2s_1-t_1\leq 2s_1$ and $n\geq 2s_1m_2\geq m_1m_2$.
\end{proof}

\begin{proof}[Proof of Lemma \ref{lem:fcr3}]
Let $D=[D_0,D_1,D_2]$, $D'=[D_0,D_1]$ and $D''=[D_0,D_2]$, then $[G_{DE_0},G_{D'E_1},G_{D''E_2}]$ is a permutation of all the columns of $G_{D_0E_0}$, $G_{D_1E_0}$, $G_{D_2E_0}$, $G_{D_0E_1}$, $G_{D_1E_1}$, $G_{D_0E_2}$, $G_{D_2E_2}$. By Lemma 1 from \cite{Harikumar1998}, it is sufficient to show that $[G_{DE_0},G_{D'E_1},G_{D''E_2}]$ has full column rank for at least one choice of $D_0,D_1,D_2,E_0,E_1,E_2$.

We complete the proof by showing that $[G_{DE_0},G_{D'E_1},G_{D''E_2}]$ has full column rank for the following choice. Let $D_0,D_1,D_2$ be chosen such that all submatrices of $\widetilde{D}_0,\widetilde{D}_1,\widetilde{D}_2$ have full rank. Let $E_0,E_1,E_2$ be chosen such that the first $2s_1s_2$ rows of $\widetilde{E}_0,\widetilde{E}_1,\widetilde{E}_2$ are
\[
\widetilde{E}_0^{(1:2s_1s_2,:)} = 
\begin{bmatrix}
I_{t_2}\otimes \mathbf{1}_{2s_1,1}\\
\hdashline
\mathbf{0}_{s_1(s_2-t_2),t_2}\\
\hdashline
\mathbf{0}_{s_1(s_2-t_2),t_2}
\end{bmatrix},\quad
\widetilde{E}_1^{(1:2s_1s_2,:)} = 
\begin{bmatrix}
\mathbf{0}_{2s_1t_2,s_2-t_2}\\
\hdashline
I_{s_2-t_2}\otimes \mathbf{1}_{s_1,1}\\
\hdashline
\mathbf{0}_{s_1(s_2-t_2),s_2-t_2}
\end{bmatrix},\quad
\widetilde{E}_2^{(1:2s_1s_2,:)} = 
\begin{bmatrix}
\mathbf{0}_{2s_1t_2,s_2-t_2}\\
\hdashline
\mathbf{0}_{s_1(s_2-t_2),s_2-t_2}\\
\hdashline
I_{s_2-t_2}\otimes \mathbf{1}_{s_1,1}
\end{bmatrix}.
\]
By the proofs of Lemmas \ref{lem:fcr1} and \ref{lem:fcr2}, $\widetilde{G}_{DE_0}^{(1:2s_1s_2,:)}$, $\widetilde{G}_{D'E_1}^{(1:2s_1s_2,:)}$ and $\widetilde{G}_{D''E_2}^{(1:2s_1s_2,:)}$ all have full column rank, and their nonzero entries are located in three disjoint row blocks. Hence $[\widetilde{G}_{DE_0},\widetilde{G}_{D'E_1},\widetilde{G}_{D''E_2}]^{(1:2s_1s_2,:)}$ has full column rank. Therefore, $[G_{DE_0},G_{D'E_1},G_{D''E_2}]$ has full column rank.
\end{proof}

\section{Proofs of the Necessary Conditions}\label{app:necessary}
\begin{proof}[Proof of Proposition \ref{pro:necessary1}]
We show that if $n<m_1+m_2-1$, then $(x_0,y_0)$ is not identifiable up to scaling. Let $\widetilde{D}_\perp\in\bbC^{n\times (n-m_1)}$ denote a matrix whose columns form a basis for the orthogonal complement of the column space of $\widetilde{D}$. Then $\widetilde{D}_\perp^*$ is an annihilator of the column space of $\widetilde{D}$, i.e., $\widetilde{D}_\perp^*\widetilde{D}=0$. Let $\widetilde{E}_\text{inv} \in\bbC^{n\times m_2}$ denote the entrywise inverse of $\widetilde{E}$:
\[
\widetilde{E}_\text{inv}^{(j,k)} =\begin{cases}
\frac{1}{\widetilde{E}^{(i,j)}} & \text{if $\widetilde{E}^{(i,j)}\neq 0$},\\
0 & \text{if $\widetilde{E}^{(i,j)}= 0$}.
\end{cases}
\]
Consider the linear operator $\calG:\bbC^{m_2}\rightarrow \bbC^{n-m_1}$ defined by
\[
\calG(w)=\widetilde{D}_\perp^*\diag(\widetilde{E}_\text{inv} w)\diag(\widetilde{E}y_0)\widetilde{D}x_0
\]
for $w\in\bbC^{m_2}$. We claim that every non-vanishing null vector of $\calG$ produces a solution to the BD problem. Indeed, if $w_1\in\calN(\calG)$ is non-vanishing, then $\diag(\widetilde{E}_\text{inv} w_1)\diag(\widetilde{E}y_0)\widetilde{D}x_0$ is annihilated by $\widetilde{D}_\perp^*$ and therefore must reside in the column space of $\widetilde{D}$. Hence, there exists $x_1\in\bbC^{m_1}$ such that
\begin{equation}
\diag(\widetilde{E}_\text{inv} w_1)\diag(\widetilde{E}y_0)\widetilde{D}x_0 = \widetilde{D}x_1.
\label{eq:3}
\end{equation}
Now, let $y_1$ denote the entrywise inverse of $w_1$. Recall that the supports of the columns of $\widetilde{E}$ are disjoint, hence $\widetilde{E}y_1$ is the entrywise inverse of $\widetilde{E}_\text{inv} w_1$. By Equation \eqref{eq:3},
\begin{align*}
\diag(\widetilde{E}y_0)\widetilde{D}x_0 &= \diag(\widetilde{E}y_1)\widetilde{D}x_1,\\
(Dx_0)\circledast (Ey_0) &= (Dx_1) \circledast (Ey_1).
\end{align*}
Hence $(x_1,y_1)$ is a solution to the BD problem where $z=(Dx_0)\circledast (Ey_0)$. This establishes the claim.

It remains to show that $\calG$ does have a non-vanishing null vector, and that the solution it produces does not coincide, up to scaling, with $(x_0,y_0)$.
Let $w_0$ denote the entrywise inverse of $y_0$, then $w_0\in\calN(\calG)$. There are $(n-m_1)$ equations in $\calG(w)=0$. If $n<m_1+m_2-1$, then $n-m_1\leq m_2-2$ and the dimension of $\calN(\calG)$ is at least $2$. Hence, there exists a vector $w_1\in\calN(\calG)$ such that $w_0,w_1$ are linearly independent. Let $\alpha$ be a complex number such that $0<|\alpha|<\frac{1}{\|y_0\|_\infty\|w_1\|_\infty}$. Then $w_0+\alpha w_1\in\calN(\calG)$ is non-vanishing, because the entries of $w_0+\alpha w_1$ satisfy that
\[
\bigl|w_0^{(j)}+\alpha w_1^{(j)}\bigr| \geq \bigl|w_0^{(j)}\bigr|-|\alpha|\bigl| w_1^{(j)}\bigr| \geq \frac{1}{\|y_0\|_\infty}-|\alpha|\|w_1\|_\infty > 0,\quad\text{for } j=1,2,\cdots,m_2.
\]
Since $\alpha\neq 0$, the null vector $w_0+\alpha w_1$ is not a scaled version of $w_0$. It produces a solution $(x_2,y_2)$ in which $y_2$ is the entrywise inverse of $w_0+\alpha w_1$ and is not a scaled version of $y_0$. Therefore, $(x_0,y_0)$ is not identifiable up to scaling.
\end{proof}

\begin{proof}[Proof of Corollary \ref{cor:necessary2}]
The vector $x_0$ is $s_1$-sparse. If we know the support of $s_1$, then the signal $u=Dx$ resides in a subspace spanned by $s_1$ columns of $D$ and the problem reduces to BD with subspace constraints. By Proposition \ref{pro:necessary1}, if $n<s_1+m_2-1$, then $(x_0,y_0)$ cannot be identified up to scaling even if the support of $x_0$ is given. Hence $(x_0,y_0)$ is not identifiable without knowing the support. Therefore, it is necessary that $n\geq s_1+m_2-1$.
\end{proof}


\end{document}